\documentclass[runningheads]{llncs}
\usepackage{graphicx}

\usepackage{amssymb}
\usepackage{graphicx}
\usepackage{extarrows}
\usepackage{amssymb,amscd}
\usepackage[misc]{ifsym}

\usepackage{tikz-cd} 
\tikzset{
	symbol/.style={
		draw=none,
		every to/.append style={
			edge node={node [sloped, allow upside down, auto=false]{$#1$}}}
	}
}

\usepackage{algorithm}
\usepackage{algorithmicx}
\usepackage{algpseudocode}

\newcommand{\rmv}[1]{}

\newcommand{\Z}{{\mathbb Z}}
\newcommand{\R}{{\mathbb R}}
\newcommand{\F}{{\mathbb F}}
\newcommand{\C}{{\mathbb C}}
\newcommand{\Q}{{\mathbb Q}}
\renewcommand{\L}{{\mathcal L}}

\renewcommand{\P}{{\mathfrak p}}

\newcommand{\z}{{\zeta}}
\newcommand{\s}{{\sigma}}
\newcommand{\disc}{{\text{disc}}}

\begin{document}
\title{On the ideal shortest vector problem over random rational primes}
\titlerunning{On the ideal shortest vector problem over random rational primes}
\author{Yanbin Pan\inst{1}
  \and Jun Xu\inst{2}
  \and Nick Wadleigh\inst{3}
  \and Qi Cheng\inst{3}  %
 }
\authorrunning{Y. Pan, J. Xu, \emph{et al.}}
%

 \institute{
 Key Laboratory of Mathematics Mechanization, 
    Academy of
    Mathematics and Systems Science,  Chinese Academy of Sciences,
    Beijing 100190, China.\\
 \email{ panyanbin@amss.ac.cn}
 \and
 State Key Laboratory of Information Security, Institute of Information Engineering 
 Chinese Academy of Sciences, Beijing 100093, China \\
 \email{ xujun@iie.ac.cn}
  \and
 School of Computer Science, University of Oklahoma, 
 Norman, OK 73019, USA.\\
\email{ ndwadleigh@gmail.com, qcheng@ou.edu}
 }

\maketitle              
\begin{abstract}
  Any non-zero ideal in a number field can be factored into a product of prime
  ideals. In this paper we report a surprising connection between the complexity
  of the shortest vector problem (SVP) of prime ideals in number fields and
  their decomposition groups. When applying the result to number fields popular
  in lattice based cryptosystems, such as power-of-two cyclotomic fields, we show that a
  majority of rational primes lie under prime ideals admitting a polynomial time
  algorithm for SVP. Although the shortest vector problem of ideal lattices
  underpins the security of the Ring-LWE cryptosystem, this work does not break
  Ring-LWE, since the security reduction is from the worst case ideal SVP to the
  average case Ring-LWE, and it is one-way.

\end{abstract}
\section{Introduction}

Due to their conjectured ability to resist quantum computer attacks, lattice-based cryptosystems 
have drawn considerable attention. In 1996, Ajtai \cite{Ajtai96} pioneered the research on worst-case to average-case reduction for the Short Integer Solution problem (SIS). In 2005, Regev \cite{Regev09} presented a worst-case to average-case (quantum) reduction for the Learning With Errors problem (LWE). SIS and LWE became two important cryptographic assumptions, and a large number of cryptographic schemes based on these two problems have been designed. However, 
the common drawback of such schemes is their limited efficiency.

To improve the efficiency of lattice-based schemes, some special algebraic structures are employed.
The first lattice-based scheme with some algebraic structure was the NTRU public key cryptosystem \cite{HoffsteinPS98}, which was introduced by Hoffstein, Pipher and Silverman in 1996. It works in the convolution ring $\mathbb{Z}[x]/(x^p-1)$ where $p$ is a prime. The cyclic nature of the ring $\mathbb{Z}[x]/(x^p-1)$ contributes to NTRU's efficiency, and makes NTRU one of the most popular schemes. Later the ring was employed in many other cryptographic primitives, such as \cite{Micciancio02,LyubashevskyM06,PeikertR06,Micciancio07,StehleS11,BernsteinC0V17}.

In 2009, Stehl\'{e} \emph{et al}. \cite{SSTX09} introduced  a
structured and more efficient variant of LWE involving the  ring $\F_p[x]/(x^N+1)$ where $N$
is a power of 2 and $p$ is a prime satisfying $p \equiv 3 \pmod{8}$. In 2010, Lyubashevsky, Peikert and Regev \cite{LPR10} presented a ring-based variant of LWE, called Ring-LWE. The hardness of problems in \cite{SSTX09,LPR10} is based on worst-case assumptions on ideal lattices. Recently, Peikert, Regev and Stephens-Davidowitz \cite{PRS17} presented a polynomial time quantum reduction from (worst-case) ideal lattice problems to Ring-LWE for any modulus and any number field. Lots of schemes  employ the ring $\mathbb{Z}[x]/(x^N+1)$ where $N$ is a power of 2, for example, NewHope \cite{NewHope}, Crystals-Kyber \cite{Kyber}, and LAC \cite{LAC} submitted to NIST's post-quantum cryptography standardization.
Although solving the ideal SVP does not necessarily break Ring-LWE,
understanding the hardness of ideal SVP is no doubt a very important first step to understand the hardness of Ring-LWE.

\subsection{Previous works}

Principal ideal lattices are a class of important ideal lattices which can be
generated by a single ring element.
There is a line of work focusing on the principal ideal SVP.
Based on \cite{CGS14,Ber14}, solving approx-SVP
problems on principal ideal lattices can be divided into the following two
steps: Step 1 is finding an ideal generator by using class group computations.
In this step, a quantum polynomial time algorithm is presented by Biasse and
Song \cite{BS16}, which is based on the work \cite{EHKS14}; a classical
subexponential time algorithm was given by Biasse, Espitau, Fouque, G\'{e}lin
and Kirchner \cite{BEFGK17}. Step 2 is shortening the ideal generator in Step 1
with the log-unit lattice. This step was analyzed by Cramer, Ducas, Peikert and
Regev \cite{CDPR16}. Then a quantum polynomial time algorithm for approx-SVP,
with a $2^{\tilde{O}(\sqrt{N})}$ approximation factor, on principal ideal
lattices in cyclotomic number fields was presented in \cite{CDPR16}.

In 2017, Cramer, Ducas and Wesolowski \cite{CDW17} extended the case of principal ideal lattices in \cite{CDPR16} to the case of a general ideal lattice in a cyclotomic ring of prime-power conductor. For approx-SVP on ideal lattices, the result in \cite{CDW17} is better than the BKZ algorithm \cite{SE94} when the approximation factor is larger than $2^{\tilde{O}(\sqrt{N})}$. Ducas, Plancon and Wesolowski \cite{DPW19} analyzed the approximation factor $2^{\tilde{O}(\sqrt{N})}$ in \cite{CDPR16,CDW17} to determine the specific dimension $N$ so that the corresponding algorithms outperform BKZ for an ideal lattice in cyclotomic number fields.
Recently, Pellet-Mary, Hanrot and Stehl\'{e} \cite{PHS19}, inspired by the algorithms in \cite{CDPR16,CDW17}, proposed an algorithm to solve approx-SVP with the approximation factor $2^{\tilde{O}(\sqrt{N})}$ in ideal lattices for all number fields,  aiming to provide trade-offs between the approximation factor  and the running time. However, there is an exponential pre-processing phase.

Inspired by Bernstein's logarithm-subfield attack \cite{Ber14}, Albrecht, Bai and Ducas \cite{ABD16} and Cheon, Jeong and Lee \cite{CJL16} independently  proposed two similar subfield attacks in 2016 against overstretched NTRU that has much larger modulus than in the NTRUEncrypt standard.  Later, Kirchner and Fouque \cite{KF17} proposed  a  variant of the subfield attacks to improve these two attacks in practice. A typical subfield attack consists of three steps: mapping the lattice to some subfield, solving the lattice problem in the subfield and finally lifting the solution to the full field.

\subsection{Our results}\label{ourres}

In this paper, we investigate the SVP for lattices corresponding to prime ideals in number fields normal over $\Q$. 
Every nonzero ideal in a Dedekind Domain can be factored uniquely into a product of prime ideals, so short
vectors in prime ideals may help us to find short vectors in general
ideals. If, in a general prime ideal $\mathfrak{p}$, we are able to efficiently find a vector with length within the
Minkowski bound for $\mathfrak{p}$, then for an ideal $\mathfrak{a}$ with few prime ideal factors, we will be able
to approximate the shortest vector in $\mathfrak{a}$ to within a factor much better than what is achieved by the LLL \cite{LLL} or
BKZ \cite{BKZ} algorithms. The most difficult step in factoring an ideal is
actually factorization of an integer (the norm of the ideal),
which can be done in polynomial time by quantum computers,
or in subexponential time by classical computers.

Consider a finite Galois extension $\mathbb{ L } \cong \Q[x]/(f(x))$ of $ \Q $,
and let $ \mathfrak{P} $ be a prime ideal  in the ring of integers
$O_\mathbb{L}$ of $\mathbb{L}$.
The subgroup of $Gal(\mathbb{L}/\Q)$ that stabilizes $ \mathfrak{P} $ set-wise is known as the
\textit{decomposition group} of $ \mathfrak{P} $. Let  $\mathbb{ K }\subset\mathbb{L}$ be the subfield fixed
by the decomposition group of $ \mathfrak{P} $.   $\mathbb{ K }$ is called the \textit{decomposition field} of $ \mathfrak{P} $.
To find a short vector in $ \mathfrak{P} $, we can
search for a short vector in the lattice $  \mathfrak{P} \cap \mathbb{K }$,
which may have smaller rank. More precisely, for a rational prime $ p $, if $pO_{\mathbb{L}}$ is
factored into a product of $ g $ prime ideals in
$O_{\mathbb{L}} $, we can reduce the problem of finding
a short vector in any of these prime ideals
to a problem of finding a short vector in
a rank-$ g $ lattice, provided that a basis of
$ O_{\mathbb{K}}  $ can be found efficiently. Equivalently, the fewer the number of irreducible factors of $ f(x) $
over $ \F_p $, the more efficiently we may solve SVP for prime ideals lying above $ p $ . One argues from general facts of algebraic number theory that the determinant of the sublattice
is not too large compared to the original lattice in order to relate Minkowski type $\lambda_1$ bounds for the two lattices.

We go on to apply the foregoing idea to the rings $ \Z[x]/(x^{2^n}+1) $, which are quite
popular in cryptography. We show that there is a hierarchy for the hardness of
SVP for these prime ideal lattices.  This arises from the observation that the decomposition groups (or their
index-two subgroups) form a chain in the subgroup lattice (See Appendices A and B). Roughly speaking, we can classify such
prime ideal lattices into $n$ distinct classes, and for a prime ideal lattice in
the $r$-th class, we can find its shortest vector by solving SVP in a
dimension-$2^r$ lattice. This suggests that the difficulty of prime ideal SVP
can change dramatically from ideal to ideal, an interesting phenomenon that has,
to our knowledge, not been pointed out in the literature.
By considering some of these classes, we
prove that a nontrivial fraction of prime ideals admit an efficient SVP
algorithm.

\begin{theorem}\label{prop:p}
	Let $ N = 2^n $, where $ n $ is a positive integer.
	Let $ {\mathfrak p} $ be a prime ideal in the ring  $ \Z[x]/(x^{N}+1) $, and suppose $\mathfrak{p}$
	contains a prime number $ p \equiv \pm 3 \pmod{8} $. Then under the coefficient embedding, the shortest vector
	in $\mathfrak{p}$
	can be found in time $ poly(N, \log p) $, and the length of the shortest vector is exactly $\sqrt{p}$.
\end{theorem}
Can we conclude from the above result that the \textit{average case} prime ideal SVP is easy? It depends how we define an average prime ideal lattice.
As prime ideals are rigid structures, changing distributions
gives us totally different complexity results.
If  prime ideals are selected
uniformly at random from the set of those prime ideals whose norms are bounded, then easy cases are rare.
Nevertheless our result does show that an average case of the prime
ideal SVP in power-of-two cyclotomic fields
is not hard,
if the prime ideals are selected
uniformly at random from the set of all prime ideals whose rational primes are less than some fixed bound.
See Subsection~\ref{averagesubs}
for details.

For general (non prime) ideals in $ \Z[x]/(x^{2^n}+1) $, we present an
algorithm to confirm that the hierarchy for the hardness of SVP also exists; that
is, we can solve SVP for a general ideal lattice by solving SVP in a $2^r$-dimensional
sublattice, for some positive integer $r$ related to the factorization of the ideal
(see Theorem \ref{mainthmg}).
Following Theorem~\ref{prop:p},  we show how to solve the SVP for ideals all
of whose prime factors lie in a certain class.
This is a special case of Theorem \ref{mainthmg}.

\begin{proposition}
	Let $ N = 2^n, $ where $ n $ is a positive integer.
	Let $ {\mathcal I} $ be an  ideal in the ring  $ \Z[x]/(x^{N}+1) $
	with prime factorization
	\[ {\mathcal I} = {\mathfrak p}_1  {\mathfrak p}_2 \cdots {\mathfrak p}_k. \]
	If each $\mathfrak{p}_i$ contains a prime integer $ \equiv  \pm 3 \pmod{8} $, the shortest vector
	in $ {\mathcal I} $  can be found in time $ poly(N, \log(\mathcal{N}({\mathcal I})) ) $.
\end{proposition}

We would like to stress that the algorithm
works by exploiting the multiplicative structure of ideals in the ring of
integers of a number field, without factoring the ideal. We regard this as the second contribution of this work, in addition to the algorithm for prime ideals.

Note that a decomposition field is a subfield of the number field. Our algorithm
can also be seen as a kind of subfield attack to solve the ideal Hermite-SVP
problem. Compared with the previous subfield attacks
\cite{Ber14,ABD16,CJL16,KF17}, the main differences are: the previous subfield
attacks use relative norm (or
trace) to map a lattice into some subfield, while we use the
intersection with the decomposition field; The approximation factor in the
previous attacks, such as \cite{ABD16}, will suffer during the lifting process,
while our lifting costs not so much; The previous attacks
\cite{ABD16,CJL16,KF17} work for NTRU with much big modulus, while the instances
amenable to our attack must satisfy the condition that the decomposition field is a proper
subfield of the number field.

We have to point out that it is still unknown how our result impacts the
security of cryptographic schemes. It does not break Ring-LWE, since the
security reduction is from the worst case ideal SVP to the average case
Ring-LWE, and it is one-way. As pointed out by \cite{Ber14}, Smart and
Vercauteren \cite{SV10} proposed an ideal lattice-based fully homomorphic
encryption scheme, which generated a prime ideal lattice as the public key. It
is enough to break the scheme by finding a short vector in the lattice. 
To improve the efficiency, they chose ideals of prime determinants, which are not
weak instances revealed by our algorithm.
Our paper provides a security justification for using such ideal lattices.
We should no doubt avoid the weak instances when we construct cryptographic schemes. In
addition, our result is a beneficial attempt to solve ideal SVP by exploiting
the algebraic structure, and it helps us understand better the hardness of
ideal SVP.

\subsection{Paper organization}

The remainder of the paper is organized as follows. In Section \ref{sec:math-prel},
we give some mathematical preliminaries. In Section \ref{sec:idealp} we prove a reduction of approx-SVP in the finite Galois extension of $\Q$.  Then in section 4 we present a reduction of SVP for prime ideal
lattices and then general ideal lattices in $\Z[\zeta_{2^{n+1}}]$.
Finally, a conclusion and some open problems are given in Section \ref{sec:con}.

\section{Mathematical preliminaries}\label{sec:math-prel}

\subsection{Lattices and some computational problems}
Lattices are discrete additive subgroups of $\R^N$. Any finite set of linearly independent vectors $b_1$, $b_2$, $\cdots$, $b_m\in \mathbb{R}^N$ generates a lattice: $$\L= \left\{\sum_{i=1}^{m}z_i b_i|\;z_i\in\mathbb{Z}\right\}.$$
Denote by $B$  the matrix  whose column vectors are the $b_i$'s.  We say $B$ is a basis (in matrix form) for $\L$; $m$ and $N$ are the rank and dimension of $\L$, respectively. Denote by  $\det(\L)$ the determinant of lattice $\L$, which is defined as the (co)volume of $\L$ in the real subspace spanned by $\L$. Note that if $m=N$,  the  determinant of $\L$ is exactly $|\det(B)|$.

The shortest vector problem (SVP), which refers to the problem of finding a shortest nonzero lattice vector in a given lattice,
is one of the most famous hard problems in lattice theory.  There are some variants of SVP that are very important for applications.
\begin{itemize}
	\item Approx-SVP: Given a lattice $\L$ and an approximation factor $\gamma\geq 1$, find
	a non-zero lattice vector of norm $\leq\gamma\cdot\lambda_1(\L)$, where  $\lambda_1(\mathcal{L})$ is the length of a shortest non-zero vector in $\L$. 
    \item Hermite-SVP: Given a rank-$N$ lattice $\L$ and an approximation factor $\gamma\geq 1$, find a non-zero lattice vector of norm $\leq\gamma\cdot\det(\L)^{\frac{1}{N}}$. Note that  Minkowski's theorem \cite{MGbook} tells us that
    $$\lambda_1(\mathcal{L}) \leq \frac{2}{\mathcal{V}_N^{1/N}}\cdot \det(\L)^{\frac{1}{N}}\leq \sqrt{N}\cdot \det(\L)^{\frac{1}{N}},$$
    where $\mathcal{V}_N$ is the volume of the $N$-dimensional ball with radius
    1.  Thus for any $\gamma\geq \sqrt{N}$,  Hermite-SVP is well-defined for all
    rank-$N$  lattices. 
    
    Since $\lambda_1(\mathcal{L})$ is usually hard to determine given a basis of $\L$, it might be very hard to verify a solution returned by an algorithm for Approx-SVP with some approximation factor. However, the solution to Hermite-SVP can be verified efficiently. Hence, 
    many algorithms, such as LLL  \cite{LLL} and BKZ \cite{BKZ}, are designed as polynomial-time
    Hermite-SVP algorithms for some exponential approximation factor. 

\end{itemize}

It is obvious that any algorithm that solves Approx-SVP with factor $\gamma$ can also solve Hermite-SVP with factor $\gamma\sqrt{N}$ by Minkowski's theorem. Furthermore, based on an idea
of Lenstra and Schnorr, Lov{\'a}sz showed that any algorithm solving Hermite-SVP with factor $\gamma$ can be used to solve Approx-SVP with factor $\gamma^2$ in polynomial time \cite{Lovasz86}. 
	
Moreover, a solution to Hermite-SVP with factor $\sqrt{N}$, that is, satisfying the Minkowski bound, is usually taken as a good enough approximation of some shortest vector in a "random" lattice. In addition, when choosing parameters for lattice-based cryptosystems in practice, such as in NewHope\cite{NewHope}, Crystals-Kyber \cite{Kyber}, and LAC \cite{LAC},  the time complexity of solving Hermite-SVP with some particular factor usually determines the concrete security of these cryptosystems. Therefore, the algorithm for Hermite-SVP is key to both solving Approx-SVP and analyzing the security of lattice-based cryptosystems.

The closest vector problem (CVP) is another famous hard problem in lattice theory.  This refers to the problem of finding a lattice vector that is closest to a given vector.

\subsection{Some basic algebraic number theory}

We will review some basic algebraic number theory in this section. More details can be found in \cite{Marcus18} or \cite{Neukirch92}
\subsubsection{Number fields}
An algebraic number $\zeta\in\mathbb{C}$ is any root of a nonzero polynomial $f(x)\in\Q[x]$ and its minimal polynomial is the unique monic irreducible $f(x)\in\Q[x]$ of minimal degree that has $\zeta$ as a root. An algebraic number is called an algebraic integer if its minimal polynomial lies in $\Z[x]$.

An algebraic number field is a finite field extension ${\mathbb{K}}$ of  $\Q$. Such a field can be obtained by adjoining a single algebraic integer $\zeta$ to $\Q$. That is, $\mathbb{K}=\Q(\zeta)$ for some algebraic integer $\zeta$. The degree $N$ of the minimal polynomial $f(x)$ of $\zeta$ is also the degree of ${\mathbb{K}}$ over $\Q$.

Denote by $O_{\mathbb{K}}$ the ring of algebraic integers in $\mathbb{ K }$. It is an integral domain and also a free $\Z$-module with rank $N$. 

For example, let $\zeta_{2^{n+1}}$ be a complex primitive
$2^{n+1}$-th root of unity, whose minimal polynomial is 
$f=x^{2^n}+1$. Then,  $\mathbb{K}=\Q(\zeta_{2^{n+1}})$ is the cyclotomic number field of order
$2^{n+1}$ with degree $2^n$.  Its ring of integers is  well known to be $\Z[\zeta_{2^{n+1}}]$.

\subsubsection{Embeddings}

A number field $\mathbb{ K }$ of degree $N$ over $\Q$ has exactly $N$ embeddings into $\C$.  Let $ \sigma_1, \sigma_2, \cdots,
\sigma_{s_1}$ be the real embeddings  from ${\mathbb{K}}$ to $\mathbb{R}$, and let $$\sigma_{s_1+1},\sigma_{s_1+2}, \cdots, \sigma_{s_1 +
	s_2}  ,$$ $$ \sigma_{s_1+s_2+1}= \overline{ \sigma_{s_1+1} },~~ \sigma_{s_1+s_2+2}= \overline{ \sigma_{s_1+2} }, ~\cdots, ~ \sigma_{s_1+2s_2}=\overline { \sigma_{s_1 +
		s_2} } $$    be the non-real embeddings  from ${\mathbb{K}}$ to $\mathbb{C}$, where $ \overline{\cdot} $ denotes
complex conjugation.

From these $\sigma_i$'s we can define the \textit{canonical embedding} $\Sigma_\mathbb{ K }$ from $\mathbb{ K }$ to $\C^{N}$:
$$\Sigma_\mathbb{K}: {\mathbb{K}} \rightarrow \C^N,~~a\mapsto (\sigma_1 (a), \sigma_2(a), \cdots, \sigma_{N}(a)).$$ 
It is known that the image of $\Sigma_\mathbb{K}$ falls into a subspace in $\C^{N}$, which is
isomorphic to $\R^N$ as an inner product space (see \cite{LPR10}).

Another important embedding from $\mathbb{ K }$ to $\R^{N}$ is the \textit{coefficient embedding,} which is most commonly used in cryptographic constructions.  This embedding depends on a choice of generator $\alpha$ for $\mathbb{ K }$: write $\mathbb{ K }= \mathbb{Q}(\alpha)$  and map $\beta = a_0+a_1\alpha+...+a_{N-1}\alpha^{N-1}$ to its coefficient vector, $C(\beta):=(a_0,a_1,...,a_{N-1})$.

If $\alpha$ may be chosen so that
$$ O_{\mathbb{K}}=\Z+\alpha \Z +\alpha^2 \Z+...+\alpha^{N-1} \Z $$
we say $ O_{\mathbb{K}}$ is  \textit{monogenic}.
In this case the coefficient embedding maps $ O_{\mathbb{K}} $ to $\Z^N$. Alternatively, via $O_\mathbb{K}\cong \Z[x]/(f(x)) $, where $f(x)$ is the minimal polynomial of $\alpha$, we may think of $C$ mapping a polynomial in $ \Z[x]/(f(x)) $ to its coefficient vector:

\[ C(a_0 + a_1 x + \cdots + a_{N-1 } x^{N-1}) = (a_0, a_1, \cdots, a_{N-1}).  \]

\subsubsection{Discriminants} 
If $\mathbb{K}\subset \mathbb{ L }$ are number fields,  the (relative) discriminant of a $\mathbb{K}$-basis $b_1, b_2, \ldots, b_N$ for $\mathbb{ L }$ is
defined by $$d_{\mathbb{ L }/\mathbb{ K }}(b_1, b_2, \ldots, b_N) = |\det(\sigma_ib_j)|^2,$$
where $\sigma_i$ varies over the $[\mathbb{ L }:\mathbb{ K }]$ embeddings $\mathbb{ L }\rightarrow \mathbb{ C }$ which fix all elements of $\mathbb{ K }$.   The discriminant $\disc({O_\mathbb{ L }/ O_\mathbb{ K }})$, also denoted by $\disc({\mathbb{ L }/\mathbb{ K }})$, is then the ideal of $O_\mathbb{ K }$ which is generated by the discriminants $d_{\mathbb{ L }/\mathbb{ K }}(b_1,b_2,\ldots,b_N)$ of all the $\mathbb{K}$-bases $b_1,b_2,\ldots,b_N$ of $\mathbb{ L }$ which are contained in $O_\mathbb{ L }$.

For any number field $\mathbb{ K }$, the (absolute) discriminant $\disc({\mathbb{ K }/\mathbb{ Q }})$ becomes the principal ideal generated by
$d(b_1,b_2,\ldots,b_N)$ for any basis $b_1,b_2,\ldots,b_N $ of the  
free $\Z$-module  $O_\mathbb{ K }$.  
In this case we just write $\disc(\mathbb{ K })$ to refer to this ideal or the unique positive integer that generates it.  In a sense made precise by the embeddings defined above, the discriminant gives a notion of the co-volume of a ring of integers in its fraction field. Specifically, the discriminant is just the square of this co-volume. 

\subsection{Ideal lattices}

	The ring of integers $ O_{\mathbb{K}} $ of $ {\mathbb{K}} $ is a free $\mathbb{Z}$-module, and any ideal $\mathcal{I}$ in $ O_{\mathbb{K}} $ is  a free $\mathbb{Z}$-submodule since $\mathbb{Z}$ is a principal ideal domain. Under the canonical embedding or the coefficient embedding, any such $\mathcal{I}$ is sent to a lattice in $ \R^{N} $.  We call this image the {\it ideal lattice} associated with $\mathcal{I}$, and we denote it also by $\mathcal{I}$.

	Under the canonical embedding $\Sigma_\mathbb{ K }$ from $\mathbb{ K }$ to
  $\C^{N}$, the co-volume (i.e. the volume of a fundamental domain) of an ideal
  lattice $\mathcal{I}$ is given by
  $N_\mathbb{K}(\mathcal{I})\sqrt{|\disc(\mathbb{ K })|}$, where $N_\mathbb{ K
  }(\mathcal{I})$ is the \textit{norm} of $\mathcal{I}$, defined as the
  cardinality of $O_\mathbb{ K }/\mathcal{I}$. Note that
  when we say the norm of a vector, it refers to the Euclidean norm rather than
  the algebra norm of an ideal.

Usually it is easier to use the canonical embedding in mathematical analysis,
and to use the coefficient embedding in cryptography. For example, under the coefficient embedding of $\Z[\zeta_{2^{n+1}}]$,
the lattice associated with the prime ideal $ \P_i=(p, f_i(\zeta_{2^{n+1}}))$
is generated by the coefficient  vectors of
the following polynomials (modulo $ x^N + 1 $)
\[  f_i, xf_i, \cdots, x^{N-1} f_i   
\ {\rm \ and\ }
p, px, \cdots, p x^{N-1},  \]
where $p$ is some rational prime, and $f_i$ is some irreducible factor of $x^{2^n}+1$ modulo $p$.
The minimum generating set should have only $ N $
vectors, which can be found by computing the
Hermite Normal Form.
	
	\subsubsection{Ideals in $\Z[\zeta_{2^{n+1}}]$}
	The cyclotomic field of order $ 2 N = 2^{n +1 } $ is widely used in cryptography.
	Its ring of integers is $\Z[\zeta_{2^{n+1}}]$, which is isomorphic to $ \Z[x]/(x^N + 1)  $.
	Its discriminant is $ 2^{n 2^n} $.
	
	Let $ p $ be a rational prime, and let \[ x^N + 1 = (f_1 f_2 \cdots f_g)^{e} \]
	be the prime factorization of $ x^N +1 $ in the polynomial ring
	$ \F_p[x] $. Then we have  $$ (p) = (\mathfrak{p}_1\mathfrak{p}_2 \cdots \mathfrak{p}_g)^{e}, $$
	where $ \mathfrak{p}_i = (p, f_i(\zeta_{2^{n+1}})) $ (here $f_i$ is any integer polynomial which projects to the $f_i$ in the above factorization).
	We say the prime ideal $ \mathfrak{p}_i  $ {\it lies over }
	the prime $ p $.
	If $ e$ is greater than $ 1 $,
	we say the prime $ p $ is {\it ramified} (in $\Z[\zeta_{2^{n+1}}]$); otherwise we say $p$ is {\it unramified.} One can verify that $ 2 $ is the only  ramified rational prime
	in the cyclotomic field of order $ 2N $, and that the
	prime ideal $ (2, \zeta_{2^{n+1}}+1) = (\zeta_{2^{n+1}}+1) $ lies above
	the ideal $ (2) $.

	We are therefore interested in the explicit factorization of the $2^{n+1}$-th cyclotomic polynomials, $x^{2^n}+1$, over $\F_p[x]$.  This is computed in \cite[Thm. 2.47 and Thm. 3.75]{GF1983} when $ p \equiv 1 \pmod 4$
	and in \cite{GF1996} when $ p \equiv 3 \pmod 4$.
	
	\begin{theorem}\label{factor1}
		Let $p \equiv 1 \pmod 4$, i.e. $p = 2^A\cdot m +1$, $A\geq 2$, $m$ odd. Denote by $U_k$ the
		set of all primitive $2^k$-th roots of unity modulo $ p $. We have
		\begin{itemize}
			\item If $n<A$, then $x^{2^n}+1$ is the product of $2^n$ irreducible linear factors over $\F_p$:
			$$x^{2^n}+1 = \prod_{u\in U_{n+1}}(x+u).$$
			\item If $n \geq A$, then  $x^{2^n}+1$ is the product of $2^{A-1}$ irreducible binomials over $\F_p$ of degree $2^{n-A+1}$:
			$$x^{2^n}+1 = \prod_{u\in U_{A}}(x^{2^{n-A+1}}+u).$$
		\end{itemize}
	\end{theorem}
	
	\begin{theorem}\label{factor3}
		Let $p \equiv 3 \pmod 4$, i.e. $p = 2^A\cdot m -1$, $A\geq 2$, $m$ odd. Denote by $D_s(x,a)$ the Dickson polynomials
		$$\sum_{i=0}^{\lfloor \frac{s}{2}\rfloor} \frac{s}{s-i}{{s-i}\choose{i}}(-a)^ix^{s-2i}$$ over $\F_p$. For $n\geq 2$, we have
		\begin{itemize}
			\item If $n<A$, then $x^{2^n}+1$ is the product of $2^{n-1}$ irreducible trinomials over $\F_p$:
			$$x^{2^n}+1 = \prod_{\gamma\in\Gamma}(x^2+\gamma x+1),$$
			where $\Gamma$ is the set of all roots of $D_{2^{n-1}}(x,1)$.
			\item If $n \geq A$, then  $x^{2^n}+1$ is the product of $2^{A-1}$ irreducible trinomials over $\F_p$ of degree $2^{n-A+1}$:
			$$x^{2^n}+1 = \prod_{\delta\in \Delta}(x^{2^{n-A+1}}+\delta x^{2^{n-A}}-1),$$
			where $\Delta$ is the set of all roots of $D_{2^{A-1}}(x,-1)$.
		\end{itemize}
	\end{theorem}

	\section{Solving Hermite-SVP for prime ideal lattices in a Galois extension}\label{sec:idealp}

In the following, we will consider solving Hermite-SVP for prime ideals of $O_\mathbb{L}$ when $\mathbb{L}$ is a finite Galois extension of $\mathbb{Q}$.  

A prime ideal $\mathfrak{p}$ in $O_\mathbb{L}$ contains a rational prime $p$, and therefore occurs as one of the prime ideals in the factorization 
	$$pO_\mathbb{L}=(\mathfrak{p}_1 \mathfrak{p}_2 \cdots \mathfrak{p}_g)^e.$$ 
Without loss of generality, we assume $\mathfrak{p}_1=\mathfrak{p}$.
	
	To find a short vector of $\mathfrak{p}_1$, we try to find a short vector in the sublattice given by the
	intersection of $\mathfrak{p}_1$ with some intermediate field between $\mathbb{Q}$ and
	$\mathbb{L}$.  Since this sublattice has smaller rank, this may lead to a more efficient algorithm than working in  $\mathbb{L}$ directly.
	
	More precisely, let $G$ be the Galois group of $\mathbb{L}$ over $\mathbb{Q}$. Recall the decomposition group, $D$, and decomposition field, $\mathbb{K}$, for  $\mathfrak{p}_1$:
	$$ D := \{\sigma\in G: \sigma(\mathfrak{p}_1)= \mathfrak{p}_1\},$$
	$$\mathbb{K}:=\{x\in\mathbb{L}:\forall \sigma\in D, \sigma(x) = x\}.$$
	Let $O_{\mathbb{K}}$ be the algebraic integer ring of $\mathbb{K}$. It is well known that the degree of $\mathbb{K}$ over $\mathbb{Q}$ is $g$ (see \cite[Thm. 28]{Marcus18}). This is our desired intermediate field, and we have the following theorem.
	\begin{theorem}\label{general}  Suppose $\mathbb{L}/\mathbb{Q}$ is a finite Galois extension with degree $N$, and suppose $\mathfrak{p}$ is a prime ideal of $O_\mathbb{L}$ lying over an unramified rational prime $p$ such that $pO_\mathbb{ L }$ has $g$ distinct prime ideal factors in $O_\mathbb{ L }$.  If $\mathbb{K}$ is the decomposition field of $\mathfrak{p}$, then a solution to Hermite-SVP with factor $\gamma$ in the sublattice $\mathfrak{c}=\mathfrak{p}\cap O_{\mathbb{K}}$ under the canonical embedding of $\mathbb{K}$ will also be  a solution to Hermite-SVP in $\mathfrak{p}$ with factor $\frac{\sqrt{N/g}}{N_{\mathbb{ K }}(\disc(\mathbb{ L }/\mathbb{ K }))^{1/(2N)}}\cdot \gamma$ ($\leq \sqrt{\frac{N}{g}}\cdot \gamma$) under the canonical embedding of $\mathbb{L}$. 
	
	In particular, when $\gamma=\sqrt{g}$, a vector in the sublattice $\mathfrak{c}$ satisfying the Minkowski bound will produce a vector in the lattice $\mathfrak{p}$ satisfying the Minkowski bound.
		 \end{theorem}
	
	\begin{proof}

			 Consider the following diagram
	\[   
	\begin{tikzcd}
	\mathfrak{p} \arrow[r,symbol=\subset]        & O_\mathbb{L} \arrow[r,symbol=\subset]                    & \mathbb{L}  \arrow[r, "\Sigma_\mathbb{L}"]                   & {\C^{N}} \\
	\mathfrak{c} \arrow[u, no head] \arrow[r,symbol=\subset]   & O_\mathbb{K} \arrow[u, no head] \arrow[r,symbol=\subset] & \mathbb{K} \arrow[r, "\Sigma_\mathbb{K}"] \arrow[u, no head] & \C^{g} \arrow[u, "\beta"]      \\
	(p)  \arrow[u, no head]  \arrow[r,symbol=\subset]         & \Z \arrow[u, no head] \arrow[r,symbol=\subset]           & \Q  \arrow[u, no head] \arrow[r,symbol=\subset]          & \C \arrow[u]                  
	\end{tikzcd}
	\]
	Here $\beta$ is chosen to be the linear map making the diagram commute.

	Note that every embedding of $\mathbb{ K }$ in $\C$ can be extended to exactly $\frac{N}{g}$ embeddings of $\mathbb{L }$ in $\C$ \cite[Thm. 50]{Marcus18}; thus $\beta$ is (up to permutation) just the linear embedding given by repeating each coordinate $N/g$ times. 
Thus for any $v\in\C^g$ we have \begin{equation}\label{eq1}
	\|\beta(v)\|=\sqrt{\frac{N}{g}}~\cdot~\|v\|.
	\end{equation}
	
	Note that the  norm of $\mathfrak{c}$ is exactly $p$ \cite[Thm. 29]{Marcus18}, so that the determinant of $\mathfrak{c}$ is $p\sqrt{|\disc(\mathbb{ K })|}$.  Thus, under the canonical embedding of $O_\mathbb{K}$ into $\mathbb{C}^g$, any solution $v_0 \in \mathfrak{c}$ to Hermite-SVP with  factor $\gamma$ satisfies
	$$ \|v_0\|\leq  \gamma\cdot p^{\frac{1}{g}} |\disc(\mathbb{K})|^{\frac{1}{2 g}}.$$

	 By Equation (\ref{eq1}) above and the fact that $\disc(\mathbb{ L })=\disc(\mathbb{ K })^{N/g}N_{\mathbb{ K }}(\disc(\mathbb{ L }/\mathbb{ K }))$ \cite[Corallary (2.10), pp. 202]{Neukirch92}, we therefore have 
	 \begin{align*}
	 	\|\beta(v_0)\|&\leq	\gamma \cdot \sqrt{\frac{N}{g}} p^{\frac{1}{g}} |\disc(\mathbb{K})|^{\frac{1}{2 g}}\\
	 	&=\gamma \cdot \frac{\sqrt{N/g}}{N_{\mathbb{ K }}(\disc(\mathbb{ L }/\mathbb{ K }))^{1/(2N)}} p^{\frac{1}{g}} |\disc(\mathbb{L})|^{\frac{1}{2 N}}\\
	 	&= \gamma \cdot \frac{\sqrt{N/g}}{N_{\mathbb{ K }}(\disc(\mathbb{ L }/\mathbb{ K }))^{1/(2N)}} (p^{\frac{N}{g}}\sqrt{ |\disc(\mathbb{L})|})^{\frac{1}{ N}}
	 \end{align*}
	 
	Note that the norm of $\mathfrak{p}$ is $p^{\frac{N}{g}}$, and thus $p^{\frac{N}{g}}\sqrt{ |\disc(\mathbb{L})|}$ is exactly the determinant of the ideal lattice $\mathfrak{p}$ under the canonical embedding of $\mathbb{L}$. Hence $v_0$ is also a solution to Hermite-SVP with factor $\frac{\sqrt{N/g}}{N_{\mathbb{ K }}(\disc(\mathbb{ L }/\mathbb{ K }))^{1/(2N)}}\cdot \gamma$.
	
	Note that $N_{\mathbb{ K }}(\disc(\mathbb{ L }/\mathbb{ K }))$ is a positive integer.  Thus $$\frac{\sqrt{N/g}}{N_{\mathbb{ K }}(\disc(\mathbb{ L }/\mathbb{ K }))^{1/(2N)}} \leq  \sqrt{\frac{N}{g}}.$$ In particular, when $\gamma=\sqrt{g}$, $\frac{\sqrt{N/g}}{N_{\mathbb{ K }}(\disc(\mathbb{ L }/\mathbb{ K }))^{1/(2N)}}\cdot\gamma \leq \sqrt{N}$ still holds. The theorem follows.\hfill $\square$
	
	\end{proof}

  \begin{remark}
    To design an algorithm from the theorem, we need to calculate the
    decomposition field from a prime ideal.
    In general this is not an easy problem. Fortunately,
    for power-of-two or prime order cyclotomic fields,
    the subfield structures have been worked out in the literature.
   Another technical problem is to compute  
    a basis for $\mathfrak{c}=\mathfrak{p}\cap O_{\mathbb{K}}$.
    This can be solved
    if we know a $ \Q $-basis of $ \mathbb{K}$.
  \end{remark}

  \begin{remark}
   How many prime ideals
  are vulnerable to this attack? In other words, given an irreducible
  polynomial over $ \Z $, how does its factoring pattern change
  over $ \F_p $
  as $ p $ varies? This is a central topic of
  class field theory when the Galois group is solvable.
  In the general case, it has been studied in  
  the famous Langlands program, where many challenging problems remain.
  The answer is well known for number fields popular in
  lattice based cryptography.
  There
  exists a set of rational primes, of positive density with non-trivial decomposition group, such that for any $p$ in
  this set, the decomposition fields of the prime ideals lying above $p$ are never
  the whole field $  \mathbb{L}$.
In this case, $\mathfrak{p}\cap O_{\mathbb{K}}$ has
	rank no more than half that of $\mathfrak{p}$, resulting in a much easier
  SVP problem.
  \end{remark}

\section{Solving SVP for ideal lattices in $\Z[\zeta_{2^{n+1}}]$}
	
In the following, we use the above idea to solve SVP for ideal lattices in $\Z[\zeta_{2^{n+1}}]$,  the ring of integers in the cyclotomic field $\Q(\zeta_{2^{n+1}})$, a field which is widely used in lattice-based cryptography. 
  The decomposition field of any prime ideal is either equal to, or a degree-$2$ subfield of, one of the following
    \[ \Q[i]\subset \Q[\zeta_8]\subset \cdots \subset \Q[\zeta_{2^n}] \subset  \Q[\zeta_{2^{ n+1 }}].  \]
 The subfields in this chain are convenient because they are monogenic and their integer rings have $\Z$-bases (powers of $\zeta_{2^{ n+1 }}$) that are mutually compatible and orthogonal under the canonical embedding.  This results in a hierarchy of complexity of prime ideal SVP problems.
  Furthermore, for a non-prime ideal $\mathcal{I}$, we can approximate the shortest
  vectors of $\mathcal{I}$ by finding short vectors in $ \mathcal{I}\cap O_{\mathbb{K}} $, where $ \mathbb{K} $ is the smallest field in the above chain containing all the  
 decomposition fields of the prime factors of $\mathcal{I}$.
  This allows us to find short vectors for
   many non-prime ideals. 
In contrast to the approximation result we achieved in the general setting of Theorem \ref{general}, an \textit{exact} SVP solution is possible in power-of-two cyclotomic fields.
We will first prove a reduction for SVP for \textit{prime} ideal lattices in $\Z[\zeta_{2^{n+1}}]$, and then we will prove a reduction for general ideals. We would like to point out that in the case of a general ideal lattice $\mathcal{I}$, we do not need to know the prime factorization of $\mathcal{I}$ to run our algorithm.

\subsection{Solving SVP for  prime ideal lattices in $\Z[\zeta_{2^{n+1}}]$} 

For simplicity we let $\zeta = \zeta_{2^{n+1}}$.  In the sequel we say goodbye to the canonical embedding and adopt the coefficient embedding $C$:
$$\mathbb{Q}( \zeta)\rightarrow \mathbb{R}^{2^n},\hspace{5mm}\sum_{i=0}^{2^n-1}a_i\zeta^i\mapsto (a_0,~a_1,...,a_{2^n-1}).$$
The coefficient embedding is widely used in cryptographic constructions.  For
power-of-two cyclotomic fields, the
two embeddings are related by  scaled-rotations, since for any $v \in \Z[\zeta_{2^{n+1}}]$ it is easy to see that 
$$\|\Sigma_\mathbb{ L }(v)\|=\sqrt{2^n}\|C(v)\|.$$   
Hence, the shortest vector under the coefficient embedding of $\mathbb{Q}( \zeta)$ must be the shortest under the canonical embedding.

	The prime $2$ is the unique ramified prime in $\mathbb{Q}(\zeta)$, and the prime ideal lying over $( 2 )$ is $ (2, \zeta+1) = (\zeta+1) $. Hence it is easy to find the shortest vector in the ideal lattice $(\zeta+1) $, and its length is  $ \sqrt{2} $.
	
	Below we consider a prime ideal lying over an odd prime and show that there
	is a hierarchy for the hardness of solving SVP for prime ideal lattices in
	$\Z[\zeta]$. Roughly speaking, we can classify all the prime ideal lattices into
	$n$ classes labeled with $1,2, \cdots, n$,
	depending on the congruence class of $ p \pmod{2^{ n+1 }}$, and for a prime ideal lattice in
	the $r$-th class, we can always find its shortest vector by solving SVP in a
	$2^r$-dimensional lattice. More precisely, we have:

	\begin{theorem}\label{mainthm}
		For any prime ideal $\mathfrak{p}= (p, f(\zeta))$ in $\Z[\zeta]$, where $p$ is an odd prime and $f(x)$ is some irreducible factor of $x^{2^n}+1$ in
		$\F_p[x]$.  Write
		$$p = \left\{
		\begin{array}{ll}
		2^A\cdot m + 1, & \hbox{if $p\equiv 1 \pmod 4$;} \\
		2^A\cdot m - 1, & \hbox{if $p\equiv 3 \pmod 4$,}
		\end{array}
		\right.
		$$
		for some odd $m$ and $A\geq 2$,  and let
		$$r = \left\{
		\begin{array}{ll}
		\min \{A-1,n\}, & \hbox{if $p\equiv 1 \pmod 4$;} \\
		\min\{A,n\}, & \hbox{if $p\equiv 3 \pmod 4$.}
		\end{array}
		\right.
		$$
		Then given an oracle that can solve SVP for $2^r$-dimensional lattices,
		a shortest nonzero vector in $\mathfrak{p}$ can be found in $\text{poly}(2^n, \log_2p)$ time with the coefficient embedding.
	\end{theorem}
	\begin{proof} It is well known that the Galois group $G$ of $\mathbb{Q}(\zeta)$ over $\mathbb{Q}$ is isomorphic to the multiplicative group
		$(\mathbb{Z}/2^{n+1}\mathbb{Z})^*$. Let $G=\{\sigma_1, \sigma_3,\cdots,\sigma_{2^{n+1}-1}\}$ where
		\begin{align*}
		\sigma_i: & \ \ \ \mathbb{Q}(\zeta)\rightarrow \ \ \mathbb{Q}(\zeta);\\
		&\ \ \ \ \  \zeta\ \ \ \mapsto \ \ \zeta^i.
		\end{align*}
		We proceed by considering two separate cases.
		
		\textbf{Case 1:} First we deal with the case when $p\equiv 1 \pmod 4$.
		The theorem is vacuously true for $n<A$.
		
		If $n\geq A$, we have $r = A-1$. By Theorem \ref{factor1}, we know that $$f(x)= x^{2^{n-A+1}}+u = x^{2^{n-r}}+u$$  for some $u\in U_{A}$. Then the prime ideal lattice $\mathfrak{p}$ can be generated by $p$ and $f(\zeta)=\zeta^{2^{n-r}}+u$. Consider the subgroup $H = \langle  \sigma_{2^{r+1}+1} \rangle$ of $G$
		generated by $\sigma_{2^{r+1}+1}$. $H$ is  a subgroup of the decomposition group of the ideal $\mathfrak{p}$ since
		$$\sigma_{2^{r+1}+1}(p)=p, \ \ \sigma_{2^{r+1}+1}(f(\zeta))=f(\zeta).$$ Note that  $\mathbb{K}=\mathbb{Q}({\zeta^{2^{n-r}}})$ is the fixed field
		of $H$ and its integer ring $O_{\mathbb{K}}$ has a $\mathbb{Z}$-basis $(1,{\zeta^{2^{n-r}}},{\zeta^{2\cdot2^{n-r}}},\cdots,{\zeta^{(2^{r}-1)\cdot2^{n-r}}} )$.
		
		Let $\mathfrak{c}=\mathfrak{p}\bigcap O_{\mathbb{K}}.$
		We claim that $ \mathfrak{p} $ is a direct sum:
		\begin{equation}\label{eq2}
		\mathfrak{p} = \bigoplus_{k=0}^{2^{n-r}-1} \zeta^k  \mathfrak{c}.
		\end{equation}
		Indeed for any
		$a\in \mathfrak{p}$, there exist integers $z_i$'s and $w_i$'s such that
		\begin{align*}
		a = &\sum_{i=0}^{2^n-1} z_i \zeta^if(\zeta) + \sum_{i=0}^{2^n-1} w_ip\zeta^i\\
		=& \sum_{k=0}^{2^{n-r}-1} \zeta^k \sum_{j=0}^{2^{r}-1} (z_{k+j\cdot2^{n-r}} \zeta^{j\cdot2^{n-r}}f(\zeta)+ w_{k+j\cdot2^{n-r}} p\zeta^{j\cdot2^{n-r}})\\
		=& \sum_{k=0}^{2^{n-r}-1} \zeta^k  \bigg((\sum_{j=0}^{2^{r}-1} z_{k+j\cdot2^{n-r}} \zeta^{j\cdot2^{n-r}})f(\zeta)+ (\sum_{j=0}^{2^{r}-1}w_{k+j\cdot2^{n-r}} \zeta^{j\cdot2^{n-r}})p\bigg).
		\end{align*}
		Let $a^{(k)}= (\sum_{j=0}^{2^{r}-1} z_{k+j\cdot2^{n-r}} \zeta^{j\cdot2^{n-r}})f(\zeta)+ (\sum_{j=0}^{2^{r}-1}w_{k+j\cdot2^{n-r}} \zeta^{j\cdot2^{n-r}})p$ for any $k$.
		Since $p\in \mathfrak{c}$ and $f(\zeta)\in
		\mathfrak{c}$, $a^{(k)}\in\mathfrak{c}$.   We have established (\ref{eq2}).

		Since multiplication by $\zeta$ is an isometry and  for $x \in \mathfrak{c}$, the coefficients of $\zeta^i x$ and $\zeta^j x$ are disjoint for $i\not= j \mod 2^{n-r}$, Equation (\ref{eq2}) implies
		$$\lambda_1(\mathfrak{p})= \lambda_1(\mathfrak{c}),$$
		and that to find the shortest vector in the ideal lattice $\mathfrak{p}$, it is enough to find the shortest vector $v$ in the ideal lattice $\mathfrak{c}$, a lattice with dimension $2^r$. Indeed $\zeta^kv$ for any $0\leq k\leq 2^{n-r}-1$ will be a shortest vector in the ideal lattice  $\mathfrak{p}$.

		\textbf{Case 2:} For the case when $p\equiv 3 \pmod 4$, everything is similar
		except that $r =A$.

		\textbf{Algorithm:} We can summarize the algorithm to solve SVP in a prime  ideal  lattice as Algorithm \ref{alg:primei}.		
		\begin{algorithm}[htb]
			\caption{Solve SVP in prime ideal lattice}
			\label{alg:primei}
			\begin{algorithmic}[1]
				\Require a prime ideal $\mathfrak{p}= (p, f(\zeta))$ in $\Z[\zeta]$, where $p$ is odd.
				\Ensure a shortest vector in the corresponding prime ideal lattice.
				\State  Compute the ideal $\mathfrak{c}$ generated by $p$ and $f(\zeta)$ in $O_{\mathbb{K}}$ where $\mathbb{K}=\mathbb{Q}({\zeta^{2^{n-r}}})$.
				\State  Find a shortest vector $v$ in the $2^r$-dimensional lattice  $\mathfrak{c}$.
				\State  Output $v$.
			\end{algorithmic}
		\end{algorithm}
		
		The most time-consuming step in Algorithm \ref{alg:primei} is Step 2 and the other steps can be done in $\text{poly}(2^n, \log_2p)$ time. \hfill$\square$
	\end{proof}
	
	\begin{remark}
			By the decomposition (\ref{eq2}) above, a similar result will hold for a prime ideal $\mathfrak{p}$ in $O_\mathbb{L}$ other than $\mathbb{Q}(\zeta)$, whenever  $O_\mathbb{ L }$ is a free $O_\mathbb{ K }$-module where $\mathbb{ K }$ is the decomposition field  of  $\mathfrak{p}$, and some $\Z$-basis of $O_\mathbb{ K }$ can be extended to the $\Z$-basis of  $O_\mathbb{L}$ that determines the coefficient embedding. 
			If we disregard the last condition that a basis of  $O_\mathbb{ K }$ extends to a basis of  $O_\mathbb{ L }$,
                 there may be a distortion of length,  depending on the basis of $O_\mathbb{ K }$, when we lift the solution from 
                 $\mathfrak{c}$ to $\mathfrak{p}$. 
                 That is, an approximation factor, which may be much larger than 1, will be involved.
               \end{remark}
	\begin{remark}
		By the remark above, solving the closest vector problem (CVP) for a prime ideal lattice can be also reduced to solving CVP in some $2^r$-dimensional sublattice.
	\end{remark}

	\subsubsection{SVP of some special prime ideals in $\Z[\zeta_{2^{n+1}}]$}
	
	Using Theorem \ref{mainthm}, we can prove Theorem~\ref{prop:p}, which shows that the SVP for prime ideals lying above some special rational primes is very easy.

\noindent \textit{Proof of Theorem~\ref{prop:p}}

	If $p \equiv -3 \pmod 8$, we may write $p=4m+1$ with odd $m$. By Theorem \ref{factor1},
	$x^{2^n}+1$ is the product of $2$ irreducible binomials over $\F_p$ of degree $2^{n-1}$:
	$x^{2^n}+1 = (x^{2^{n-1}}+u_1)\cdot(x^{2^{n-1}}+u_2),$
	where $u_i$ satisfies $u_i^2 \equiv -1 \pmod p$.
	
	For any prime ideal $(p, \zeta^{2^{n-1}}+u_i)$ over $(p)$,  by the proof of Theorem \ref{mainthm},  the shortest vector can be found by solving the $2$-dimensional lattice $\mathcal{L}_i$ generated by
	$
	\left(
	\begin{array}{cc}
	u_i & 1 \\
	-1 & u_i \\
	p & 0 \\
	0 & p \\
	\end{array}
	\right).
	$
	Note that $(-1, u_i)\equiv u_i\cdot(u_i, 1) \pmod p$ and $(0,p) = p\cdot(u_i,1)-u_i\cdot(p,0)$. The generator matrix can be reduced to the basis of $\mathcal{L}_i$ as
	$
	\left(
	\begin{array}{cc}
	u_i & 1 \\
	p & 0 \\
	\end{array}
	\right),
	$  which is exactly the Hermite Normal Form of the lattice basis.

For any vector $v \in \mathcal{L}_i$, there exists an integer vector $(z_1,z_2)$ such that $v= (z_1,z_2)\left(
\begin{array}{cc}
u_i & 1 \\
p & 0 \\
\end{array}
\right)=(z_1 u_i+z_2 p,z_1)$. Note that
\begin{align*}
\|v\|^2 = (z_1 u_i+z_2 p)^2+z_1^2 = z_1^2(u_i^2+1) +z_2^2p^2+2pz_1z_2u_i \equiv 0 \pmod p.
\end{align*}
Then for the nonzero shortest vector $v$, we have $0<\|v\|^2<\frac{4}{\pi}\cdot p<2p$ (by Minkowski's Theorem \cite{MGbook}) and $\|v\|^2 \equiv 0 \pmod p$, which implies that
$\|v\|^2=p$.

In case $p \equiv 3 \pmod 8$, we may write $p=4m-1$ with odd $m$. By Theorem \ref{factor3},  then  $x^{2^n}+1$ is the product of $2$ irreducible binomials over $\F_p$ of degree $2^{n-1}$:
$x^{2^n}+1 = (x^{2^{n-1}}+\delta_1x^{2^{n-2}}-1)\cdot (x^{2^{n-1}}+\delta_2x^{2^{n-2}}-1),$
where $\delta_i$ satisfies $\delta_i^2\equiv -2 \pmod p$ since the Dickson polynomial is $D_2(x,-1)=X^2+2$.

For any prime ideal $(p, \zeta^{2^{n-1}}+ \delta_i \zeta^{2^{n-2}}-1)$ over $(p)$, we similarly consider the shortest vector in $\mathcal{L}_i$ generated by
$$\left(
\begin{array}{cccc}
-1 & \delta_i & 1 & 0 \\
0 & -1 & \delta_i & 1 \\
-1 & 0 & -1 & \delta_i \\
-\delta_i & -1 & 0 & -1 \\
p & 0 & 0 & 0 \\
0 & p & 0 & 0 \\
0 & 0 & p & 0 \\
0 & 0 & 0 & p \\
\end{array}
\right).
$$
Similarly, we can easily get the basis for $\mathcal{L}_i$ in the Hermite Normal Form
$$
\left(
\begin{array}{cccc}
0 & -1 & \delta_i & 1 \\
-1 & \delta_i & 1 & 0 \\
0 & p & 0 & 0 \\
p & 0 & 0 & 0 \\
\end{array}
\right),
$$
and prove that for any vector $v\in\mathcal{L}$,
$$\|v\|^2 \equiv 0 \pmod p.$$
For the shortest vector $v$, by Minkowski's Theorem, we know $0<\|v\|^2\leq \frac{4\sqrt{2}}{\pi}<2p$, which implies that $\|v\|^2 =p$. By Theorem \ref{mainthm}, the proposition follows. \hfill$\square$

	\subsection{SVP average-case hardness for prime ideals in $\Z[\zeta]$}\label{averagesubs}
	Precisely defining the average-case hardness of SVP for a prime ideal lattice in $\Z[\zeta]$ requires specifying
	a distribution. We consider the following three distributions.
	\subsubsection{The first distribution.}
	To select a random prime ideal, one fixes a large $M$, uniformly randomly selects a prime number in the set 
	\[ \{  p {\rm \ is\ a\ prime}: p< M  \},\]
	and then uniformly randomly selects a prime ideal lying over $ p $.
	This process provides a reasonable distribution among prime ideals,
	since every prime ideal in the ring of integers of
	$ \Q[x]/(f(x)) $  is of the form $(p, g(x))$, where $p$ is
	a prime number and $g(x)$ is an irreducible factor of $ f(x) $ over $ \F_p[x] $.
	Since roughly half of all primes $p\leq M$ satisfy $ p \equiv \pm 3 \pmod{8} $, according
	to Dirichlet's theorem on arithmetic progressions, at least half of all such $p$ have the property that the ideals lying over $p$
	admit an efficient algorithm
	for SVP.
	
	\subsubsection{The second distribution.}
	
	Again fixing a large $M$,  we might alternatively
	select a prime ideal uniformly at random from the set
	\[ \{ \mathfrak{p} {\rm \ prime\ ideal} : p \in \mathfrak{p}, p {\rm \ is\ a\ prime}, p< M  \}.\]
	In this case, a non-negligible fraction of prime ideals admit efficient SVP
	algorithm.  More precisely, we have

	\begin{proposition} Under the distribution above, a random prime ideal of $\Z[\zeta]$ admits an efficient SVP algorithm with probability at least $\frac{1}{1+2^{n-1}}$.
		
	\end{proposition}
	\begin{proof}For simplicity, we disregard the single prime ideal lying over 2.
		Note that for $p=8k\pm 3$, there are exactly two prime ideals over $p$, and, by  Theorem~\ref{prop:p},  the SVP for the corresponding ideal lattices is easy. For $p=8k\pm 1$, there are at most $2^n$ prime ideals lying over $p$, by Theorems \ref{factor1} and \ref{factor3}. Then by Dirichlet's prime number theorem, even if we only count the prime ideals lying over $p=8k\pm 3$, the fraction of easy instances is at least $\frac{1}{1+2^{n-1}}$. \hfill$\square$
	\end{proof}

	\subsubsection{The third distribution.}
	The third distribution is more common in mathematics.  Namely, after fixing a large $M$,
	we select uniformly at random a prime ideal from the set
	\[ \{ \mathfrak{p} {\rm \ prime\ ideal} : \mathcal{N}(\mathfrak{p})< M  \},\]
	where $\mathcal{N}(\mathfrak{p})$ is the norm of the ideal $\mathfrak{p}$.
	
	By Theorem \ref{mainthm}, SVP for a prime ideal lattice $\mathfrak{p}$ reduces to SVP for a $2^r$-dimensional sub-lattice $\mathfrak{c}$, where $r$ is as defined in the statement of Theorem \ref{mainthm}.  Note that our algorithm will not improve matters if $r=n$, that is, if $p$ splits completely in $\mathbb{Q}(\zeta)$, or equivalently if $\mathcal{N}(\mathfrak{p})=p$.  By Chebotarev's density theorem \cite{Che1926}, there are about $\frac{M}{2^{n}\log M}$ rational primes which split in $\mathbb{Q}(\zeta)$ and hence $\frac{M}{\log M}$ prime ideals lying above those primes, for which our algorithm cannot provide a reduction for SVP.
	
	If our algorithm is to provide a reduction, the prime ideal under study must lie over a rational prime $p$ with $p\leq \sqrt{M}$, since $\mathcal{N}(\mathfrak{p})=p^{f}< M $ where $f$ is some integer greater than 1. Hence there are at most $\sqrt{M}$ such primes and hence at most $2^{n-1}\sqrt{M}$ prime ideals for which our algorithm provides a reduction.
	
	Under such a distribution, therefore, the density of the easy instances for our algorithm is at most $\frac{2^{n-1}\log M}{\sqrt{M}}$, which goes to zero when $M$ tends to infinity.
	
	\begin{remark}
From a cryptographic perspective,  there seems to be no construction relying
on the average hardness of ideal SVP in ideals following one of the two first
distributions above. However, our algorithm  reveals the concrete reason why we should avoid such distributions in the cryptographic constructions although it seems very easy to sample according to the two distributions.

\end{remark}
	
	\subsection{Solving SVP for a general ideal lattice in $\mathbb{Z}[\zeta_{2^{n+1}}]$} \label{sec:idealg}

	For simplicity, we let $\zeta = \zeta_{2^{n+1}}$. We will show that 
	even for a general ideal lattice $\mathcal{I} \subset\Z[\zeta]$, there is a similar hierarchy for the hardness of SVP. We would like to stress that although the following theorem refers to the prime factorization of $\mathcal{I}$, the resulting algorithm does not require it. 
		
	\begin{theorem}\label{mainthmg}
		Let $\mathcal{I}$ be a nonzero ideal of $\Z[\zeta]$ with prime factorization
		$$ \mathcal{I} = \mathfrak{p}_1\cdot \mathfrak{p}_2\cdots \mathfrak{p}_t,$$
		where $\mathfrak{p}_i =(f_i(\zeta),p_i)$ for rational primes $p_i$, and where the $\mathfrak{p}_i$ are  not necessarily distinct. Write
		$p_i = 2^{A_i}\cdot m_i + 1$ when $p_i\equiv 1 \pmod 4$ and $p_i = 2^{A_i}\cdot m_i - 1$ when $p_i\equiv 3 \pmod 4$  with odd $m_i$, and let
		$r = \max \{r_i\},$ where
		$$r_i = \left\{
		\begin{array}{ll}
		\min \{A_i-1,n\}, & \hbox{if $p_i\equiv 1 \pmod 4$;} \\
		\min\{A_i,n\}, & \hbox{if $p_i\equiv 3 \pmod 4$;}\\
		n, &\hbox{if $p_i =2$.}
		\end{array}
		\right.
		$$
		Then the shortest vector in the ideal lattice $\L$ corresponding to $\mathcal{I}$  can be solved via solving SVP in a $2^r$-dimensional lattice.
	\end{theorem}
	
	\begin{proof}
		If $r = n$, then the theorem follows simply.

		If $r<n$, W.L.O.G., we assume $r=r_1$. Following the proof of  Theorem \ref{mainthm}, denote the Galois group $G=\{\sigma_1, \sigma_3,\cdots,\sigma_{2^{n+1}-1}\}$ of $\mathbb{Q}(\zeta)$ over $\mathbb{Q}$, where $\sigma_i(\zeta) = \zeta^i$.  Consider the subgroup $H = \langle \sigma_{2^{r+1}+1} \rangle$ of $G$ generated by $\sigma_{2^{r+1}+1}$.  For any $\tau\in H$ and every prime ideal $\mathfrak{p}_i= (p_i, f_i(\zeta))$, we have $\tau(\mathfrak{p}_i)=\mathfrak{p}_i$ since
		$\sigma_{2^{r+1}+1}(p_i)=p_i, \ \ \sigma_{2^{r+1}+1}(f_i(\zeta))=f_i(\zeta).$  Note that  $\mathbb{K}=\mathbb{Q}({\zeta^{2^{n-r}}})$ is the fixed field
		of $H$ and its integer ring $O_{\mathbb{K}}$ has a $\mathbb{Z}$-basis $(1,{\zeta^{2^{n-r}}},{\zeta^{2\cdot2^{n-r}}},\cdots,{\zeta^{(2^{r}-1)\cdot2^{n-r}}} )$.
		
		Let $\mathfrak{c}=\mathcal{I}\bigcap O_{\mathbb{K}}.$
		We claim that for any
		$a\in \mathcal{I}$, there exist $a^{(k)}\in \mathfrak{c}$ for $0\leq k< 2^{n-r}$, such that
		\[ a= \sum_{k=0}^{2^{n-r}-1} \zeta^k a^{(k)}.
		\]
		We proceed by induction. When $t=1$ the above claim holds by  Theorem \ref{mainthm}. Suppose the claim holds for $t-1$.  Then setting $\mathcal{I} = \mathfrak{p}_1\cdot \mathfrak{p}_2\cdots \mathfrak{p}_t$,  and
		$\overline{\mathcal{I}} = \mathfrak{p}_1\cdot \mathfrak{p}_2\cdots \mathfrak{p}_{t-1}$,  we have $\mathcal{I} = \overline{\mathcal{I}}\cdot \mathfrak{p}_t$. For any $a\in \mathcal{I}$, we can write $a=\sum x_i y_i$ where $x_i\in \overline{\mathcal{I}}$ and $y_i\in \mathfrak{p}_t$.
		It suffices to show that for any $xy$, where $x\in \overline{\mathcal{I}}$ and $y\in \mathfrak{p}_t$, there exist $b^{(k)}\in \mathcal{I}\bigcap O_{\mathbb{K}}$ for $0\leq k< 2^{n-r}$, such that $xy = \sum_{k=0}^{2^{n-r}-1} \zeta^k b^{(k)}$.
		
		By the induction assumption, there exist $x^{(i)}\in \overline{\mathcal{I}}\bigcap O_{\mathbb{K}}$ for $0\leq i< 2^{n-r}$ such that $x = \sum_{i=0}^{2^{n-r}-1} \zeta^i x^{(i)}$, and  there exist
		$y^{(j)}\in \mathfrak{p}_t\bigcap O_{\mathbb{K}}$ for $0\leq j< 2^{n-r}$ such that $y = \sum_{j=0}^{2^{n-r}-1} \zeta^j y^{(j)}$.  Hence, we have
		\begin{align*}
		xy = &\sum_{i=0}^{2^{n-r}-1}\sum_{j=0}^{2^{n-r}-1} \zeta^{i+j} x^{(i)}y^{(j)}\\
		=& \sum_{k=0}^{2^{n-r}-1} \zeta^k \sum_{i+j=k} x^{(i)}y^{(j)} +  \sum_{k=2^{n-r}}^{2\cdot2^{n-r}-2} \zeta^k \sum_{i+j=k} x^{(i)}y^{(j)}\\
		=& \sum_{k=0}^{2^{n-r}-1} \zeta^k \sum_{i+j=k} x^{(i)}y^{(j)} +  \sum_{k=0}^{2^{n-r}-2} \zeta^k \sum_{i+j=k+2^{n-r}} \zeta^{2^{n-r}} x^{(i)}y^{(j)}\\
		=& \sum_{k=0}^{2^{n-r}-2} \zeta^k (\sum_{i+j=k} x^{(i)}y^{(j)}+\sum_{i+j=k+2^{n-r}} \zeta^{2^{n-r}} x^{(i)}y^{(j)}) + \zeta^{2^{n-r}-1}\sum_{i+j=2^{n-r}-1} x^{(i)}y^{(j)}.
		\end{align*}
		Let $b^{(k)}= \sum_{i+j=k} x^{(i)}y^{(j)}+\sum_{i+j=k+2^{n-r}} \zeta^{2^{n-r}} x^{(i)}y^{(j)}$ for any $0\leq k\leq 2^{n-r}-2$  and $b^{(2^{n-r}-1)}=\sum_{i+j=2^{n-r}-1} x^{(i)}y^{(j)}$.
		We have that $b^{(k)}\in\mathcal{I}\bigcap O_{\mathbb{K}}$ for $0\leq k< 2^{n-r}$. Hence, for any
		$a\in \mathcal{I}$, there exist $a^{(k)}\in \mathfrak{c}$ for $0\leq k< 2^{n-r}$, such that
		$a=\sum_{k=0}^{2^{n-r}-1} \zeta^k a^{(k)}.$
		
		As in the proof of Theorem \ref{mainthm}, we can show that $\lambda_1(\mathcal{I})=\lambda_1(\mathfrak{c})$ and any nonzero shortest vector in $\mathfrak{c}$ will yield $2^{n-r}$ nonzero shortest vectors in $\mathcal{I}$.\hfill $\square$

	\end{proof}
	
	We would like to point out that in some cases, the $r$ in Theorem \ref{mainthmg}
	can be improved. Consider the case when $n\geq 3$ and
	$\mathcal{I}=(2,\zeta-1)^2=(2,\zeta^2+1)$. We need to solve SVP in a $2^n$-dimensional lattice by Theorem \ref{mainthmg}. However, using the intermediate field $\mathbb{Q}(\zeta^2)$ as in the proof of Theorem \ref{mainthmg},  we can find a shortest vector by solving SVP in a $2^{n-1}$-dimensional lattice.

	Furthermore, since for any
	$a\in \mathcal{I}$, there exist $a^{(k)}\in \mathfrak{c}$ for $0\leq k< 2^{n-r}$, such that
	$
	a=\sum_{k=0}^{2^{n-r}-1} \zeta^k a^{(k)},
	$
	we conclude that if $(b^{(i)})_{0\leq i<2^r}$ is a basis of the ideal lattice $\mathfrak{c}$, then $(\zeta^j b^{(i)})_{0\leq i<2^r, 0\leq j<2^{n-r}}$ is a basis of the ideal lattice $\mathcal{I}$. Denote by $\mathcal{L}_j$ the lattice generated by $(\zeta^j b^{(i)})_{0\leq i<2^r}$. Then we have that the ideal lattice $\mathcal{I}$ has an orthogonal decomposition:
	$\L_0\oplus \L_1\oplus\cdots\oplus\L_{2^{n-r}-1}.$

	In fact, for any $\bar{r}$, let $\mathfrak{c}=\mathcal{I}\bigcap O_{\mathbb{K}}$ where  $\mathbb{K}=\mathbb{Q}({\zeta^{2^{n-\bar{r}}}})$. For any basis $(b^{(i)})_{0\leq i<2^{\bar{r}}}$ of the ideal lattice $\mathfrak{c}$, if $(\zeta^j b^{(i)})_{0\leq i<2^{\bar{r}}, 0\leq j<2^{n-\bar{r}}}$ is a basis of the ideal lattice $\mathcal{I}$ (meaning that the ideal lattice  $\mathcal{I}$ has an orthogonal decomposition), then the shortest vector in $\mathfrak{c}$ is also a shortest vector in $\mathcal{I}$. Hence we have the following algorithm to solve SVP for a general ideal in $\mathbb{Z}[\zeta]$ without knowing the prime factorization of the ideal.
	\begin{algorithm}
		\caption{Solve SVP in general ideal lattice}
		\label{alg:generali}
		\begin{algorithmic}[1]
			\Require an ideal $\mathcal{I}$;
			\Ensure a shortest vector in the corresponding ideal lattice $\L$.
			\For{$\bar{r}=1$ to n}
			\State Compute a basis $(b^{(i)})_{0\leq i<2^{\bar{r}}}$  of the ideal lattice $\mathfrak{c}=\mathcal{I}\bigcap O_{\mathbb{K}}$, where  $\mathbb{K}=\mathbb{Q}({\zeta^{2^{n-\bar{r}}}})$.
			\If{$(\zeta^j b^{(i)})_{0\leq i<2^{\bar{r}}, 0\leq j<2^{n-{\bar{r}}}}$ is exactly a basis of ideal lattice $\mathcal{I}$}
			\State Find a shortest vector $v$ in the $2^{\bar{r}}$-dimensional lattice  $\mathfrak{c}$;
			\State  Output $v$.
			\EndIf
			\EndFor
			
		\end{algorithmic}
	\end{algorithm}
	
	Note that Step 2 can be done efficiently by computing the intersection of the lattices $\mathcal{I}$ and $O_{\mathbb{K}}$ under the coefficient embedding.

	\begin{remark}
		By the proof of Theorem \ref{mainthmg}, solving the closest vector problem (CVP) for a general ideal lattice can also be reduced to solving CVP in some $2^r$-dimensional lattice.
	\end{remark}

	\section{Conclusion and open problems}\label{sec:con}
	
	We have investigated the SVP of prime ideal lattices in  the finite Galois extension of $\mathbb{Q}$, 
  and designed an algorithm exploiting the subfield structure of such fields to solve Hermite-SVP for prime ideal lattices.	
	For	the power-of-two cyclotomic fields, 
	 we obtained an efficient algorithm for solving SVP in many ideal lattices, either prime or non-prime ideals.
	We also determined the length of the shortest vector of those prime ideals
	lying over rational primes congruent to $  \pm 3\pmod{8} $.
	It is an interesting problem to study the length of the shortest vectors
	in other prime ideals.
	The worst case hardness of prime ideal lattice SVP for
	power-of-two cyclotomic fields is also left open.

  \subsubsection*{Acknowledgements.} We thank the anonymous referees for their
  valuable suggestions on how to improve this paper. This work is supported by
  National Key Research and Development Program of China (No. 2020YFA0712300,
  2018YFA0704705),  National Natural Science Foundation of China (No.
  62032009, 61732021, 61572490) for Y. Pan and J. Xu, and National Science Foundation of USA
  (CCF-1900820) for N. Wadleigh and Q. Cheng.

\bibliographystyle{splncs04}
\bibliography{ilsvp}

	\appendix

	\section{The subfields of $\mathbb{Q}(\zeta_{2^n})$}
	Now we sketch the subfield lattice of $\mathbb{Q}(\zeta_{2^{n+1}})$. Consider the three subfields $$\mathbb{Q}(\zeta_{2^{n+1}}+\zeta_{2^{n+1}}^{-1}),~\mathbb{Q}(\zeta_{2^{n}}),~\mathbb{Q}(\zeta_{2^{n+1}}-\zeta_{2^{n+1}}^{-1}).$$
	
	\noindent First we claim $\mathbb{Q}(\zeta_{2^{n+1}})$ is degree two over each. On the one hand, all are proper subfields since $\mathbb{Q}(\zeta_{2^{n+1}}+\zeta_{2^{n+1}}^{-1})$ is contained in the fixed field of the automorphism $\zeta_{2^{n+1}}\mapsto\zeta_{2^{n+1}}^{-1}$, and $\mathbb{Q}(\zeta_{2^{n+1}}-\zeta_{2^{n+1}}^{-1})$
	is in the fixed field of the automorphism $\zeta_{2^{n+1}}\mapsto-\zeta_{2^{n+1}}^{-1}$. 
	On the other hand, $\zeta_{2^{n+1}}$ is a root of the quadratic polynomials $x^2-(\zeta_{2^{n+1}}+\zeta_{2^{n+1}}^{-1})x+1\in \mathbb{Q}(\zeta_{2^n}+\zeta_{2^{n+1}}^{-1})[x]$ and $~~x^2-(\zeta_{2^{n+1}}-\zeta_{2^{n+1}}^{-1})x-1\in \mathbb{Q}(\zeta_{2^{n+1}}-\zeta_{2^{n+1}}^{-1})[x].$  
	
	Moreover, since the involutions $$\zeta_{2^{n+1}}\mapsto\zeta_{2^{n+1}}^{-1}, ~\zeta_{2^{n+1}}\mapsto \zeta_{2^{n+1}}^{2^{n-1}+1},~\zeta_{2^{n+1}}\mapsto-\zeta_{2^{n+1}}^{-1}$$ are distinct, these three subfields are distinct.  
	Finally it is routine to sketch the subgroup lattice of
	$\mathbb{Z}_2\oplus\mathbb{Z}_{2^{n-1}}\cong(\mathbb{Z}/2^{n+1}\mathbb{Z})^*\cong \mathrm{Gal}(\mathbb{Q}(\zeta_{2^{n+1}})/\mathbb{Q})$:	
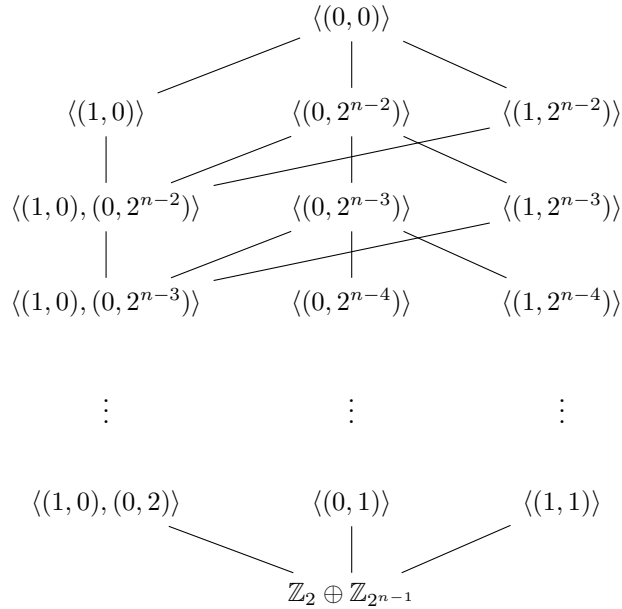
	$$
	\begin{tikzcd}
	 {}                                                                                            & {\langle (0,0)\rangle}                                                                            &                                                 \\
   {\langle(1,0)\rangle} \arrow[overlay, ru, no head]                                                     & {\langle(0,2^{n-2})\rangle} \arrow[u, no head]                                                    & {\langle(1,2^{n-2})\rangle} \arrow[lu, no head] \\
  {\langle(1,0),(0,2^{n-2})\rangle} \arrow[u, no head] \arrow[ru, no head] \arrow[rru, no head] & {\langle(0,2^{n-3})\rangle} \arrow[u, no head]                                                    & {\langle(1,2^{n-3})\rangle} \arrow[lu, no head] \\
{\langle(1,0),(0,2^{n-3})\rangle} \arrow[u, no head] \arrow[ru, no head] \arrow[rru, no head] & {\langle(0,2^{n-4})\rangle} \arrow[u, no head]                                                             & {\langle(1,2^{n-4})\rangle} \arrow[lu, no head] \\
   \vdots                                                                                        & \vdots                                                                                            & \vdots                                          \\
  {\langle(1,0),(0,2)\rangle}                                                                   & {\langle(0,1)\rangle}                                                                             & {\langle(1,1)\rangle}                           \\                                                                                            & \mathbb{Z}_2\oplus\mathbb{Z}_{2^{n-1}} \arrow[u, no head] \arrow[ru, no head] \arrow[lu, no head] &                                                
	\end{tikzcd}	$$
		
	\noindent Here all lines indicate extensions of index two. Combining these facts we have the  subfield lattice for $\mathbb{Q}(\zeta_{2^n})$:
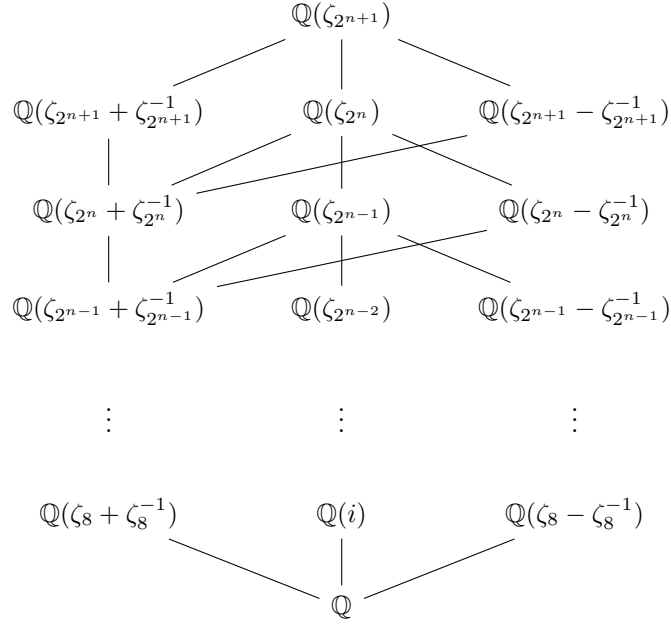
	$$
	\begin{tikzcd}
	& \mathbb{Q}(\zeta_{2^{n+1}})  &       {}                                                             \\
	\mathbb{Q}(\zeta_{2^{n+1}}+\zeta_{2^{n+1}}^{-1}) \arrow[rru, phantom] \arrow[ru, no head] & \mathbb{Q}(\zeta_{2^{n}}) \arrow[u, no head]   & \mathbb{Q}(\zeta_{2^{n+1}}-\zeta_{2^{n+1}}^{-1}) \arrow[lu, no head]         \\
	\mathbb{Q}(\zeta_{2^{n}}+\zeta_{2^{n}}^{-1}) \arrow[rru, no head] \arrow[ru, no head] \arrow[u, no head] & \mathbb{Q}(\zeta_{2^{n-1}}) \arrow[u, no head]                        & \mathbb{Q}(\zeta_{2^{n}}-\zeta_{2^{n}}^{-1}) \arrow[lu, no head]  \\
	\mathbb{Q}(\zeta_{2^{n-1}}+\zeta_{2^{n-1}}^{-1}) \arrow[u, no head] \arrow[rru, no head] \arrow[ru, no head] & \mathbb{Q}(\zeta_{2^{n-2}}) \arrow[u, no head]                        & \mathbb{Q}(\zeta_{2^{n-1}}-\zeta_{2^{n-1}}^{-1}) \arrow[lu, no head] \\
	\vdots                                                                                                       & \vdots                                                                & \vdots                                                              \\
	\mathbb{Q}(\zeta_{8}+\zeta_{8}^{-1})                                                                         & \mathbb{Q}(i)                                                         & \mathbb{Q}(\zeta_{8}-\zeta_{8}^{-1})                                \\
	& \mathbb{Q} \arrow[lu, no head] \arrow[u, no head] \arrow[ru, no head] &                                                                      &  
	\end{tikzcd}		
	$$
	\noindent where all lines indicate extensions of order two.

	\section{Decomposition groups and fixed fields}
	Let $\zeta = \zeta_{2^{n+1}}$, $p$ a rational prime with $p\equiv 3 \pmod{ 4 }$,
	$A$ the natural number with $2^A|| p+1$, and let $\frak p$ be a prime ideal in
	$\Z[\zeta]$ containing $p$. Then $$\frak{p} = (p,~ \z^{2^{n-A+1}}+ \delta
	\z^{2^{n-A}}-1)$$ for some $\delta \in \Z$. Let $\sigma \in Aut(\Q(\z)/\Q)$
	be the automorphism of $\Q(\z)$ with $\z\mapsto \z^{-2^A-1}$.
	Then we have
	\begin{align*}
	\s \frak{p}&= (p, ~ \s(\z)^{2^{n-A+1}}+ \delta
	\s(\z)^{2^{n-A}}-1)\\ 
	&=(p, ~\z^{2^{n-A+1}(-2^{A}-1)}+ \delta \z^{2^{n-A}(-2^{A}-1)}-1)\\
	&=(p, ~\z^{-2^{n+1}}\z^{-2^{n-A+1}}+ \delta \z^{-2^{n}}\z^{-2^{n-A}}-1)\\
	&=(p, ~\z^{-2^{n-A+1}}- \delta \z^{-2^{n-A}}-1)\\ 
	&=(p,~-\z^{-2^{n-A+1}}\cdot(\z^{2^{n-A+1}}+ \delta \z^{2^{n-A}}-1))\\
	&=\frak p.
	\end{align*}
	We have used the fact that $\z$ is a unit in $\Z[\z].$

	Since $\z\mapsto \z^{-1}$ is an involution, the order of $\sigma$ is the order
	of $\z\mapsto \z^{2^A+1}$ (denoted by $ \s' $ ) which is the multiplicative order of $2^A+1$ in
	$(\Z/2^{n+1}\Z)^\ast$. We claim that, for $A\geq 2$, this
	order is $2^{n+1-A}$: First note that for $k\equiv 1\pmod { 4 }$,
	$$\mathrm{ord}_{(\Z/{2^{n+1}}\Z )^*}(k)=2^m$$ if and only if $2^{n+1}|| k^{2^m}-1$. This
	fact follows easily from the identity $$k^{2^{g+1}}-1= (k^{2^{g}}-1)
	(k^{2^{g}}+1)$$ and the fact that for $k=2^A+1$, we have $2|| (k^{2^{g}}+1)$.
	Now, that the multiplicative order of $2^A+1$ is $2^{n+1-A}$ follows from an
	induction argument using the above identity.
	
	The preceding two paragraphs prove that $\s$ lies in the decomposition group of
	$\frak p$ and that $\s$ has order $2^{n+1-A}$. It follows from a standard result
	in the theory of number fields that the decomposition group of $\frak p$ has
	order $2^{n+1-A}$. Thus $\langle \s \rangle$ is precisely the decomposition
	group of $\frak p$. Now recall the subfield/subgroup lattice for $\Q(\z)/\Q$ and
	its Galois group $\mathbb{Z}_{2^{n+1}}^*$. A simple computation shows that
	$\s $ fixes $\z^{2^{n-A}}-\z^{-2^{n-A} }$. But from the subfield
	lattice we can see that $$[\Q(\z): \Q(\z^{2^{n-A}}-\z^{-2^{n-A} })] = 2^{n+1-A}= |\langle
	\s \rangle |.$$ Thus $\Q(\z^{2^{n-A}}-\z^{-2^{n-A} })$ is precisely this fixed field.
	
	A similar, in fact easier, analysis can be carried out for $p\equiv 1 \pmod{4}$.
	In this case $$\frak{p} = (p,~ \z^{2^{n-A+1}}-u)$$ for some $u\in \Z$ and
	$2^A||p-1$. Then it is seen that $\s'$ fixes $\frak{p}$.
	As in the $3 \pmod{4}$ case, we know from a general result of algebraic number
	theory that the decomposition group of $\frak p$ has order $2^{n+1-A}$, which
	matches the order of $\sigma'$ (computed above). We see that
	$\Q(\z^{2^{n+1-A} })$ is contained in the fixed field of $\sigma'$, and again, by
	looking at the subfield lattice to find $[\Q(\z): \Q(\z^{2^{n+1-A}})] = 2^{n+1-A}$,
	we see that $\Q(\z^{2^{n+1-A} })$ is precisely the fixed field of the decomposition
	group of $\frak p$.

%
%
%

\end{document}